\documentclass[conference]{IEEEtran}

\usepackage{color}

\usepackage[caption=false]{subfig}
\usepackage{graphicx}
\usepackage{amsmath, amsthm, amssymb}

\newtheorem{definition}{Definition}
\newtheorem{theorem}{Theorem}
\newtheorem{remark}{Remark}
\newtheorem{corollary}{Corollary}

\renewcommand{\P}{\mathbb{P}}
\newcommand{\E}{\mathbb{E}}

\newcommand{\Rbb}{\mathbb{R}}

\newcommand{\Bcal}{\mathcal{B}}

\newcommand{\set}[1]{\left\{#1\right\}}

\title{An Infinite Dimensional Model for\\ A Many Server Priority Queue}

\author{ %
  \IEEEauthorblockN{Neal Master, Zhengyuan Zhou, and Nicholas Bambos}%
  \IEEEauthorblockA{Department of Electrical Engineering, Stanford University\\
    Stanford, CA 94305\\
    \texttt{\{nmaster, zyzhou, bambos\}@stanford.edu}} }

\begin{document}

\maketitle
\thispagestyle{empty}
\pagestyle{empty}

\begin{abstract}
  We consider a Markovian many server queueing system in which
  customers are preemptively scheduled according to exogenously
  assigned priority levels. The priority levels are randomly assigned
  from a continuous probability measure rather than a discrete one and
  hence, the queue is modeled by an infinite dimensional stochastic
  process. We analyze the equilibrium behavior of the system and
  provide several results. We derive the Radon-Nikodym derivative
  (with respect to Lebesgue measure) of the measure that describes the
  average distribution of customer priority levels in the system; we
  provide a formula for the expected sojourn time of a customer as a
  function of his priority level; and we provide a formula for the
  expected waiting time of a customer as a function of his priority
  level. We verify our theoretical analysis with discrete-event
  simulations. We discuss how each of our results generalizes previous
  work on infinite dimensional models for single server priority
  queues.
\end{abstract}

\section{Introduction}
Priority queueing models arise in several applications. In packet
switched communication networks, priority levels are used to deliver
differentiated levels of quality of service, e.g. \cite{Shin_2001,
  Semeria_2001}. In emergency medicine, priority queueing models are
used to study triage policies, e.g. \cite{Green_2006}. Priority
queueing models are also used in financial engineering to model order
books in which limit orders are given priority for matching with other
orders based on their price and time of arrival
\cite{Cont_2010}. Because priority queueing models are useful is so
many domains, several priority queueing models exist; see
\cite{Jaiswal_Book} for a standard reference on stochastic priority
queueing.

In this paper, we formulate and analyze an $M/M/c$ priority queueing
model in which priority levels are drawn from a
\emph{continuum}. Unlikely previous models that allow for finitely
many priority levels, e.g. \cite{Helly_1962, Burke_1962}, our model
requires an infinite dimensional state process. Consequently, standard
Markov chain techiques that apply to finitely many priority levels,
e.g. \cite{White_1958}, do not apply. Our recent previous work
\cite{Master_ACC_2017} also considered a continuous distribution of
priority levels but only for the single server case. The current paper
generalizes the results in our previous work \cite{Master_ACC_2017} by
extending the results to a many server queue.

The idea of using a continuous distribution for randomly assigning
priority levels was also recently proposed as a scheduling mechanism
for the $M/G/1$ queue \cite{Haviv_2016}. Although the preemptive
priority scheduling mechanism is the same for both our work and the
work in \cite{Haviv_2016}, a major difference is that our work (and
our previous work \cite{Master_ACC_2017}) provides a characterization
of the distribution of customer priority levels in the system in
equilibrium. In contrast, \cite{Haviv_2016} focuses more on the effect
of the randomized scheduling on the overall population. Another major
difference is that we do not assume that the system is stable. Our
current work considers a system with many servers and is hence
distinct from both \cite{Master_ACC_2017} and \cite{Haviv_2016}.

Because of the complexity that arises due to having a continuum of
priority levels, we opt to simplify other aspects of the model.  We
note that all customers in our model experience the same service rate
regardless of their priority level. This differs from other priority
queueing models, e.g. \cite{Takacs_1964}, and restricts our attention
to models in which priority levels only impact scheduling and not
service rate. We also focus on preemptive scheduling as in
\cite{Chang_1965} rather than non-preemtive scheduling as in
\cite{Kapadia_1984}. By focusing on preemptive scheduling we know that
the customer who is being serviced is always the customer with the
highest priority. Both of these assumptions (uniform service rate and
preemptive scheduling) were also exploited in \cite{Master_ACC_2017}
and \cite{Haviv_2016}.

We note that the use of infinite dimensional stochastic processes is
itself not novel to queueing. Measure-valued processes have been used
to study the earliest-deadline-first discipline \cite{Doytchinov_2001}
and the processor-sharing discipline \cite{Gromoll_2004}, as well as
many server \cite{Kaspi_2013} and infinite server models
\cite{Reed_2015}. In these contexts, the state of the system varies
continuously as the dynamic properties of the jobs change. In our
model, the priority levels are static and so the state only changes at
arrival and departure events. Consequently, our model is substantially
more tractable. Indeed, while these other works focus on diffusion
approximations, we will only present exact results.

With this motivation and background in mind, the remainder of the
paper is organized as follows. In Section~\ref{sec:model} we fully
describe our model and discuss different choices for the state. In
Section~\ref{sec:theory} we analyze the steady state behavior of the
system. We compute the measure that tells us the average distribution
of customer priority levels in the system. We derive formulae for the
expected sojourn and waiting times of a customer as a function of his
priority level. We note how these results generalize our previous work
\cite{Master_ACC_2017}. In Section~\ref{sec:sim} we provide some
simulation results that verify our analytical results. In
Section~\ref{sec:future} we discuss potential future work and we
conclude in Section~\ref{sec:conclusions}.

\section{Model Formulation\label{sec:model}}
In this section we formally describe our model and explain our
modeling assumptions. We highlight the fact that certain seemingly
limiting assumptions are actually without loss of generality. We
present three state representations and explain their
equivalence. This model is very similar to the model from our previous
work \cite{Master_ACC_2017}; the key difference is that here we allow
for more than one server.

We consider an infinite buffer queue with $c$ servers. Customers
arrive according to a Poisson process with rate $\alpha >
0$. Customers have independent and identically distributed (IID)
service times that are exponentially distributed. Since time can be
scaled arbitrarily, we assume that the service times have unit mean.
Therefore, the load is $\alpha / c$. Customers are also assigned IID
priority levels that are uniformly distributed on the unit
interval. The priority levels are independent of all other random
quantities in the model. Customers are scheduled preemptively
according to their priorities. When there are at most $c$ customers in
the system, each customer is assigned to a server; when there are more
than $c$ customers in the system, the $c$ with the highest priority
levels are assigned to the $c$ servers while the rest wait. When a new
customer arrives and no servers are available, he may immediately
preempt the lowest priority customer who is in service. The preempted
customer waits in the buffer.  In summary, we have an $M/M/c$ queue
(not necessarily stable) in which customers are preemptively scheduled
according to exogenously assigned IID $U([0, 1])$ priority levels.

Note that customers are scheduled based on their relative order rather
than their absolute value and consequently, the fact that the priority
levels are drawn from $U([0, 1])$ (as opposed to some other
distribution) is actually without loss of generality. Because the
scheduling decisions only depend on the relative order of the priority
levels, the dynamics would be unchanged if the priorities were
transformed by any order-preserving map. In particular, suppose we
want to consider priority levels that are drawn from some other
distribution with cumulative distribution function (CDF)
$F(\cdot)$. Consider two distinct customers $i$ and $j$ with priority
levels $p_i$ and $p_j$ drawn from $U([0,1]$. Consider the transformed
priority levels $\tilde p_i = F^{-1}(p_i)$ and $\tilde p_j =
F^{-1}(p_j)$ where $F^{-1}(\cdot)$ is the quantile function associated
with $F(\cdot)$:
\begin{equation}
  F^{-1}(p) = \inf \set{x \in \Rbb : p \leq F(x)}
\end{equation}
If $p_i > p_j$ then $\tilde p_i \geq \tilde p_j$ and we also have that
$\tilde p_i$ and $\tilde p_j$ are distributed according to the CDF
$F(\cdot)$ \cite[Theorem~2.1]{Devroye_1986}. So if $F(\cdot)$ is
strictly increasing then using $\tilde p_i$ and $\tilde p_j$ yields
the same scheduling dynamics as using $p_i$ and $p_j$. If $F(\cdot)$
is not strictly increasing, then with non-zero probability we could
have $\tilde p_i = \tilde p_j$. However, in this situation customers
$i$ and $j$ are indistinguishable and these ties can be broken in an
arbitrary fashion, e.g. randomly. Consequently, our model encompasses
arbitrary distributions of priority levels. For simplicity, we will
focus having priority levels drawn from $U([0, 1])$.

We also note that because of the memorylessness property of the
exponential distribution, if a customer is preempted then his residual
service time is still exponentially distributed with unit mean. As a
result, any choice for the state does not need to include the residual
service time of each customer in the system, merely the priority level
of each customer. Since the priority levels are drawn from a
continuum, almost surely no two customers will have the same priority.
Consequently, the state needs to encode the unique priority level of
each customer in the system. We find it convenient to encode this list
of priority levels as a point measure on $[0, 1]$. Let $\Bcal([0, 1])$
be the $\sigma$-algebra of Borel sets on $[0, 1]$. Given $B \in
\Bcal([0, 1])$ let $x_t(B)$ be the number of customers in the system
at time $t$ with priority levels contained in $B$. To write this
symbolically, let $\delta_z$ denote a Dirac measure at $z \in [0,
1]$. If there are $N$ customers in the system at time $t$ and their
priority levels are $\set{p_1, \hdots, p_N} \subset [0, 1]$, then
$x_t$ can be written as a sum of Dirac measures:
\begin{equation}
  x_t = \sum_{i=1}^N \delta_{p_i}
\end{equation}

Now consider the (non-normalized) CDF or the complementary CDF:
\begin{equation}
  X_t(p) = x_t([0, p]),\quad \bar X_t(p) = x_t((p, 1])
\end{equation}
These two function-valued stochastic processes are actually equivalent
to the measure-valued process defined above. The equivalence follows
from the fact that $\set{[0, p] : p \in [0, 1]}$ and $\set{(p, 1] : p
  \in [0, 1]}$ each form $\pi$-systems that generate
$\Bcal([0,1])$. We know that $x_t(\cdot)$ is finite because it is a
counting measure.  Hence, an elementary application the
$\pi$-$\lambda$ Theorem shows that $\set{X_t(p) : p \in [0, 1]}$ and
$\set{\bar X_t(p) : p \in [0,1]}$ each uniquely define
$x_t(\cdot)$. The definitions of $\pi$-systems and $\lambda$-systems
along with the method of uniquely extending a measure from a
$\pi$-system to a $\sigma$-algebra are standard; see
\cite[Chapter~3]{Billingsley} for details.

\section{Some Theoretical Results\label{sec:theory}}
We now analyze the equilibrium behavior of the system. First we
characterize the steady state distribution of $\bar X_t(p)$ for each
$p \in [0, 1]$. We provide a corollary that partially characterizes
the steady state distribution of $x_t(\cdot)$. We then provide
formulae for the expected sojourn time and the expected waiting time
of a customer as functions of its priority level. Each of these
results generalizes our previous results regarding the single server
case \cite{Master_ACC_2017}. As in the previous section, we rely on
standard results regarding the extension of measures from
$\pi$-systems to $\sigma$-algebras \cite[Chapter~3]{Billingsley}.

\begin{theorem}
  \label{thrm:X_bar}
  Fix any $p \in [0, 1]$, $\bar X_t(p)$ converges weakly to a random
  variable $\bar X(p)$. If $(1-p)\alpha < c$, then $\bar X(p)$ has the
  following probability mass function (PMF) on the non-negative
  integers:
  \begin{align}
    &\P(\bar X(p) = k)\\
    &= \left\{
      \begin{array}{ll}
        \left[\sum_{i=0}^{c -1} \frac{((1-p)\alpha)^i}{i!} + \frac{((1-p)\alpha)^c}{c! \times (1 - (1-p)(\alpha/c))}\right]^{-1}&, k = 0\\
        \P(\bar X(p) = 0) \times  \frac{((1-p)\alpha)^k}{k!}&, 1 \leq k \leq c\\
        \P(\bar X(p) = 0) \times \frac{((1-p)\alpha)^k}{c! \times c^{k - c}} &, k > c
      \end{array}
    \right.\nonumber
  \end{align}
  If $(1-p)\alpha \geq c$, then $\bar X(p) = \infty$ almost surely.
\end{theorem}
\begin{proof}
  As in \cite{Master_ACC_2017}, the key is to notice that because of
  the preemptive scheduling, the customers with priority levels in
  $(p, 1]$ are not affected in any way by customers with priority
  levels in $[0, p]$. Moreover, because the priority levels are
  independent of the inter-arrival times, the customers with priority
  levels in $(p, 1]$ arrive according to a Poisson process with rate
  $(1-p)\alpha$. As a result, $\bar X_t(p)$ is stochastically equivalent
  to the population in an $M/M/c$ queue with unit rate servers and
  arrival rate $(1-p)\alpha$. As a result, $\bar X_t(p)$ converges
  weakly to a random variable with the given PMF
  \cite[Chapter~3]{Kleinrock}.  Because there is no upper bound on
  $\alpha$, it is possible that $(1-p)\alpha \geq c$. In this case, the
  equivalent $M/M/c$ queue is not stable and hence $\bar X_t(p)$
  diverges to infinity.
\end{proof}

\begin{remark}
  When $c = 1$ and $(1-p)\alpha < c$ we have that
  \begin{equation}
    \P(\bar X(p) = k) = (1 - (1-p)\alpha) ((1-p)\alpha)^k
  \end{equation}
  for all non-negative integers $k$. In other words, $\bar X(p)$ is a
  geometric random variable on the non-negative integers with mean
  $(1-p)\alpha / (1 - (1-p)\alpha)$. Hence, this result generalizes our
  previous work \cite[Theorem~1]{Master_ACC_2017}.
\end{remark}

\begin{definition}
  For convenience, we define 
  \begin{align}
    P_0(p) %
    &= \P(\bar X(p) = 0) \nonumber\\
    &= \left[\sum_{i=0}^{c -1} \frac{((1-p)\alpha)^i}{i!} + \frac{((1-p)\alpha)^c}{c! \times (1 - (1-p)(\alpha/c))}\right]^{-1}
  \end{align}
  when $(1 - p)\alpha < c$. We also define 
  \begin{align}
    p_0(p) %
    &= -\frac{d}{dp}P_0(p)\\
    &= -P_0(p)^2\Bigg[\sum_{i=1}^{c-1} \frac{i (1-p)^{i-1}\alpha^i}{i!} %
    + \frac{c(1-p)^{c-1}\alpha^c}{c!(1 - (1-p)(\alpha/c))} \nonumber\\
    &\quad\quad\quad\quad + \frac{(1-p)^c \alpha^{c+1}}{c\times c! (1 - (1-p)(\alpha/c))^2}\Bigg]\nonumber
  \end{align}
\end{definition}

\begin{corollary}
  \label{cor:x}
  Fix $B \in \Bcal([0, 1])$. Then $x_t(B)$ converges weakly to a
  random variable $x(B)$ with mean
  \begin{equation}
    \mu(B) = \E[x(B)] = \int_B m(p) dp
  \end{equation}
  where 
  \begin{align}
    m(p) &= \alpha +\\
    &\frac{\alpha^{c+1}}{c\times c!}\Bigg[%
    \frac{(c + 1)(1-p)^c P_0(p) + (1-p)^{c+1} p_0(p)}{(1 - (1-p)(\alpha/c))^2} \nonumber\\
    &\quad\quad\quad\quad+%
    \frac{2(1-p)^{c+1} P_0(p) (\alpha/c)}{(1 - (1-p)(\alpha/c))^3}\Bigg]\nonumber
  \end{align}
  when $(1-p)\alpha < c$ and $m(p) = \infty$ otherwise.
\end{corollary}
\begin{proof}
  We can use the PMF from the previous theorem to show that
  \begin{align}
    \E[\bar X(p)] = (1-p)\alpha + \frac{\alpha^{c+1} (1-p)^{c+1} P_0(p)}{c \times c! \times (1 - (1-p)(\alpha/c))^2}.
  \end{align}
  This is the average number of customers in an $M/M/c$ queue with
  arrival rate $(1-p)\alpha$ and unit service rate.  Therefore, if $B
  = [a, b]$ for some $0 \leq a < b \leq 1$, then $x_t(B) = \bar X_t(a)
  - \bar X_t(b)$. Since $x_t([a, b])$ converges weakly to $\bar X(a) -
  \bar X(b)$, performing the integration gives us the same result
  subtracting $\E[\bar X(b)]$ from $\E[\bar X(a)]$. Indeed, note that
  for $p$ such that $(1-p)\alpha < c$,
  \begin{equation}
    m(p) = -\frac{d}{dp} \E[\bar X(p)].
  \end{equation}
  Now note that intervals of this form are a $\pi$-system that
  generates $\Bcal([0, 1])$. Consequently, if $\alpha < c$ then
  $\mu([0,1]) < \infty$ and so this defines a unique measure on
  $\Bcal([0, 1])$. On the other hand, if $\alpha \geq c$, we can still
  extend the measure from the $\pi$-system to $\Bcal([0, 1])$, but
  uniqueness is no longer guaranteed. However, we can apply the same
  reasoning as above to define a unique measure on
  $\Bcal([1-c/\alpha,1])$ where $\mu(\cdot)$ is finite. The fact that
  $\mu(B) = \infty$ for any $B$ such that $B \cap [0,1 - c/\alpha]$
  has non-zero Lebesgue measure follows from the instability argument
  in the previous theorem.  Hence, regardless of the value of
  $\alpha$, we can conclude that the expression for the mean
  equilibrium behavior of $x_t(B)$ holds for any $B \in \Bcal([0,1])$.
\end{proof}

\begin{remark}
  When $c = 1$ and $(1-p)\alpha < c$ we have that
  \begin{equation}
    m(p) = \frac{\alpha}{(1 - (1-p)\alpha)^2}
  \end{equation}
  so the corollary generalizes the results in our previous work
  \cite{Master_ACC_2017}.
\end{remark}

Because service can be preempted and hence customers can enter service
multiple times, we formally define the terms ``sojourn time'' and
``waiting time''. In particular, we note that the amount of time a
customer spends in service before being preempted is considered
waiting. We used the same definitions in our prior
work~\cite{Master_ACC_2017}.
\begin{definition}
  A customer's sojourn time is the amount of time from when the customer arrives
  to when it departs after completing service.
\end{definition}
\begin{definition}
  A customer's waiting time is the amount of time from when the customer arrives
  to the beginning of the last time the customer enters service.
\end{definition}

\begin{theorem}
  Fix any $p \in [0, 1]$ and let $s(p)$ be the expected sojourn time
  for a customer with priority $p$ in steady state. Then if
  $(1-p)\alpha < c$ then
  \begin{equation}
    s(p) = \frac{1}{\alpha}m(p)
  \end{equation}
  and if $(1-p)\alpha \geq c$ then $s(p) = \infty$.
\end{theorem}
\begin{proof}
  The case for which $s(p) = \infty$ follows trivially from the
  instability argument in Theorem~\ref{thrm:X_bar}.

  For the nontrivial case, we first consider $\bar S(p)$, the average
  sojourn time for all customers with priority levels in $(p, 1]$. The
  law of total probability tells us that
  \begin{equation}
    \bar S(p) = \int_p^1 s(q) \frac{1}{1 - p} dq.
  \end{equation}
  Now we apply Little's Law \cite{Little_1961}. Customers with
  priority levels in $(p, 1]$ arrive at a rate $(1-p)\alpha$ so we
  have that
  \begin{equation}
    \E[\bar X(p)] = (1-p)\alpha \bar S(p) = \alpha \int_p^1 s(q) dq.
  \end{equation}
  The corollary gives us a formula for $\E[\bar X(p)]$:
  \begin{equation}
    \int_p^1 s(q)dq = \frac{1}{\alpha}\E[\bar X(p)] = \frac{1}{\alpha}\int_p^1 m(q)dq
  \end{equation}
  Since this holds for any $p$, we have that $s(p) = m(p)/\alpha$.
\end{proof}

\begin{corollary}
  Fix any $p \in [0, 1]$ and let $w(p)$ be the expected waiting time
  for a customer with priority $p$ in steady state to receive
  service. If $(1-p)\alpha < c$ then
  \begin{equation}
    w(p) = s(p) - 1 = \frac{1}{\alpha}m(p) - 1
  \end{equation}
  and if $(1-p)\alpha \geq c$ then $w(p) = \infty$.
\end{corollary}
\begin{proof}
  The sojourn time is the sum of the waiting time and the service
  time. Since we have a unit service rate, we merely subtract 1 from
  $s(p)$ to get $w(p)$.
\end{proof}

\begin{remark}
  We note that the relationships between $m(\cdot)$, $s(\cdot)$, and
  $w(\cdot)$ are the same as they were in the single server case
  \cite{Master_ACC_2017}. Consequently, the previous theorem and
  corollary generalize the results from our previous work.
\end{remark}

\begin{remark}
  If $\alpha \geq c$ then $m(\cdot)$ (and hence both $s(\cdot)$ and
  $w(\cdot)$) exhibit a bifurcation, i.e. a qualititative change in
  behavior, at
  \begin{equation}
    p^* = 1 - \frac{c}{\alpha}.
  \end{equation}
  It is intuitive that when the system is overloaded, lower priority
  customers will be ignored so that higher priority customers can be
  served. The quantity $p^*$ makes this intuition precise: when the
  system is overloaded, customers with priority levels in $[0, p^*]$
  will have infinite expected waiting times while customers in $(p^*,
  1]$ will have finite expected waiting times.
\end{remark}

\begin{remark}
  The aforementioned birfurcation makes the case of $\alpha = c$
  particularly interesting. We know that when $\alpha = c$ the $M/M/c$
  is unstable. However, in this case $p^* = 0$ so all customers with
  priority levels in $(0, 1]$ have a finite sojourn time while only
  customers with priority levels equal to zero have infinite sojourn
  times. This seems a bit paradoxical: the queue is unstable but
  almost every customer has a finite sojourn time. This
  counterintuitive result arises because $\alpha = c$ is the critical
  point between a stable $M/M/c$ queue and an unstable $M/M/c$ queue.
\end{remark}

\begin{remark}
  The previous remarks highlight the fact that this infinite
  dimensional priority scheduling scheme can be used to ``partially
  stabilize'' an unstable single class queueing system in the
  following sense. If we have a single class $M/M/c$ system with
  $\alpha \geq c$ that is scheduled in either a last-come-first-serve
  (LCFS) or first-come-first-serve (FCFS) manner, then we know that
  the overall population of the queue will be unstable and we cannot
  provide any guarantee of reasonable service to any of the
  customers. If we instead randomly assign priority levels to arriving
  customers and schedule preemptively according to these priority
  levels, then we can guarantee that $c/\alpha$ of the customers can
  expect to have finite waiting times. Moreover, upon arrival we can
  say with certainty exactly which customers will have this guarantee.
\end{remark}

\section{Simulation Verification\label{sec:sim}}
In this section, we report the results of two discrete event
simulations of the system: one with $\alpha < c$ and one with $\alpha
\geq c$.  In both cases, we use the simulated data to estimate
$m(\cdot)$, $s(\cdot)$, and $w(\cdot)$. In general, we see that the
estimates match our theoretical results, thus supporting our analysis.

\subsection{Estimation Methods}
For each of the functions that we estimate, we first get local
estimates and we then linearly interpolate to estimate the entire
function. The details for each function are outlined below and are the
same as in our previous work~\cite{Master_ACC_2017}. For all
functions, we assume a discretization of $0 < \delta < 1$ with an
integer $N_\delta = \delta^{-1}$.

We compute our estimate of $m(\cdot)$, which we denote $\hat
m(\cdot)$, as follows:
\begin{enumerate}
\item Because ``Poisson Arrivals See Time
  Averages''~\cite{Wolff_1982}, we record $x_t(\cdot)$ as observed
  immediately before each arrival.
\item For $p_i \in \set{\delta/2 + i\delta }_{i=0}^{N_\delta - 1}$, we
  average the number of customers with priority levels in the
  half-open interval $[p_i - \delta/2, p_i + \delta/2)$ across our
  observations. We scale this average by $N_\delta$ to get $\hat
  m(p_i)$.
\item We linearly interpolate $\set{\hat m(p_i)}_{i=0}^{N_\delta - 1}$
  to get $\hat m(\cdot)$.
\end{enumerate}

We compute our estimate of $s(\cdot)$, which we denote $\hat
s(\cdot)$, in a similar fashion:
\begin{enumerate}
\item We record the arrival time, the departure time, and the priority
  level of each customer. If a customer does not depart in the time
  horizon, then his departure time is infinite.
\item For $p_i \in \set{\delta/2 + i\delta}_{i=0}^{N_\delta - 1}$, we
  average the sojourn times for customers with priority levels in the
  half-open interval $[p_i - \delta/2, p_i + \delta/2)$. This gives us
  $\hat s(p_i)$.
\item We linearly interpolate $\set{\hat s(p_i)}_{i=0}^{N_\delta - 1}$
  to get $\hat s(\cdot)$.
\end{enumerate}

We compute our estimate of $w(\cdot)$, which we denote $\hat
w(\cdot)$, in a similar fashion:
\begin{enumerate}
\item We record the arrival time, the last time that the customer
  enters service before departing, and the priority level of each
  customer. If the customer never departs then the departure time is
  infinite.
\item For $p_i \in \set{\delta/2 + i\delta}_{i=0}^{N_\delta - 1}$, we
  average the waiting times for customers with priority levels in the
  half-open interval $[p_i - \delta/2, p_i + \delta/2)$. This gives us
  $\hat w(p_i)$.
\item We linearly interpolate $\set{\hat w(p_i)}_{i=0}^{N_\delta - 1}$
  to get $\hat w(\cdot)$.
\end{enumerate}

\subsection{Estimation Results}
We use $\delta = 0.05$ and a time horizon of $T = 2 \times 10^3$. We
fix $c = 2$ servers and consider two values of $\alpha$. When $\alpha
= 1.5$ we have a stable system and when $\alpha = 5.0$ we have an
unstable system.

First we consider the stable case in which $m(\cdot)$, $s(\cdot)$, and
$w(\cdot)$ are finite. The results are plotted in
Fig.~\ref{fig:stable}. Though a bit ``noisy'', the estimates generally
agree with our theoretical analysis. Moreover, we see that the
estimates have roughly the same shape and merely differ by constant
factors. This confirms our previous analysis regarding the mean
equilibrium behavior of $x_t(\cdot)$, the expected sojourn time, and
the expected waiting time.

\begin{figure}
  \subfloat[$\hat m(\cdot)$ vs. $m(\cdot)$\label{fig:m}]{%
    \includegraphics[width=0.45\textwidth]{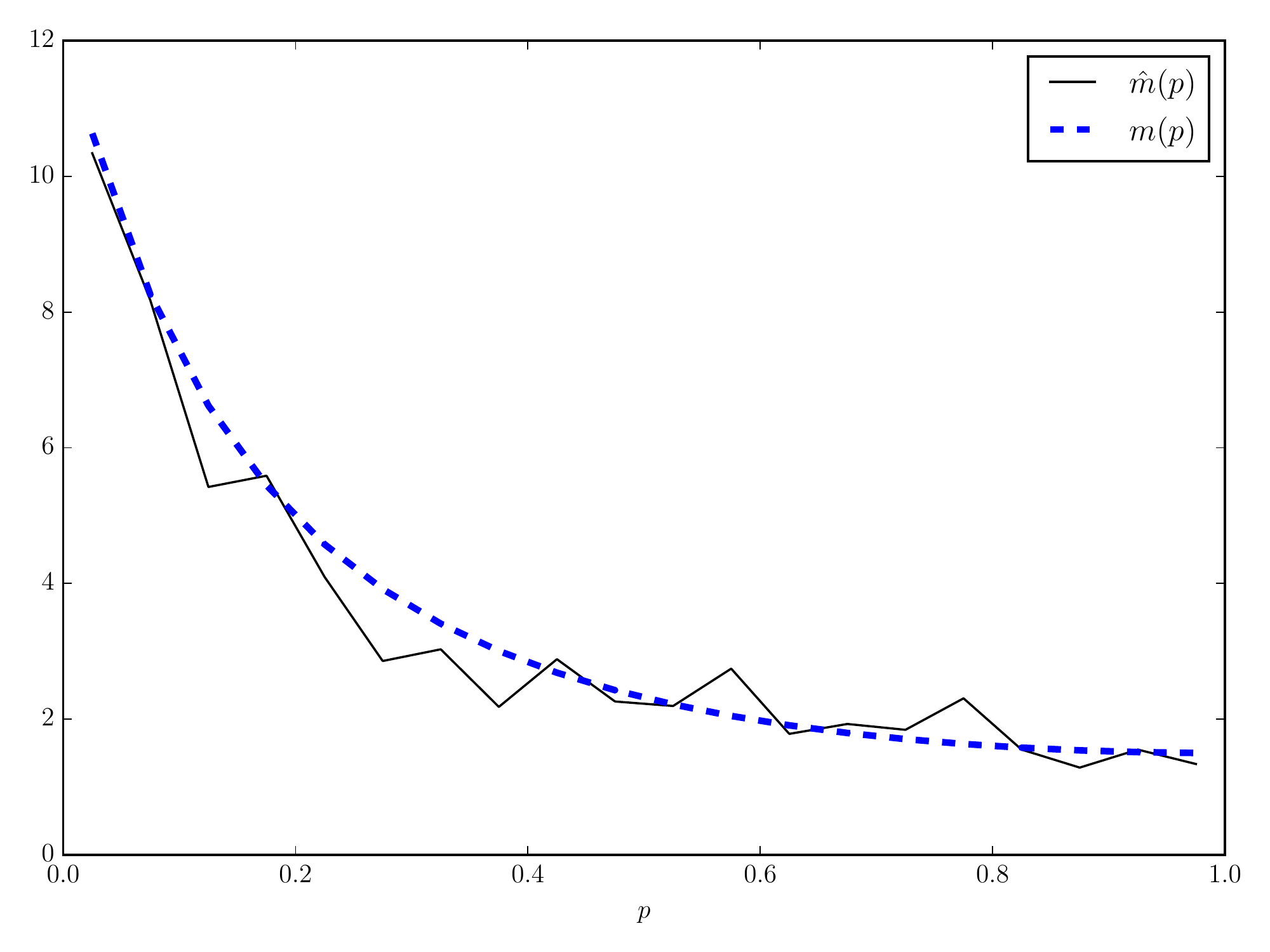}}

  \subfloat[$\hat s(\cdot)$ vs. $s(\cdot)$\label{fig:s}]{%
    \includegraphics[width=0.45\textwidth]{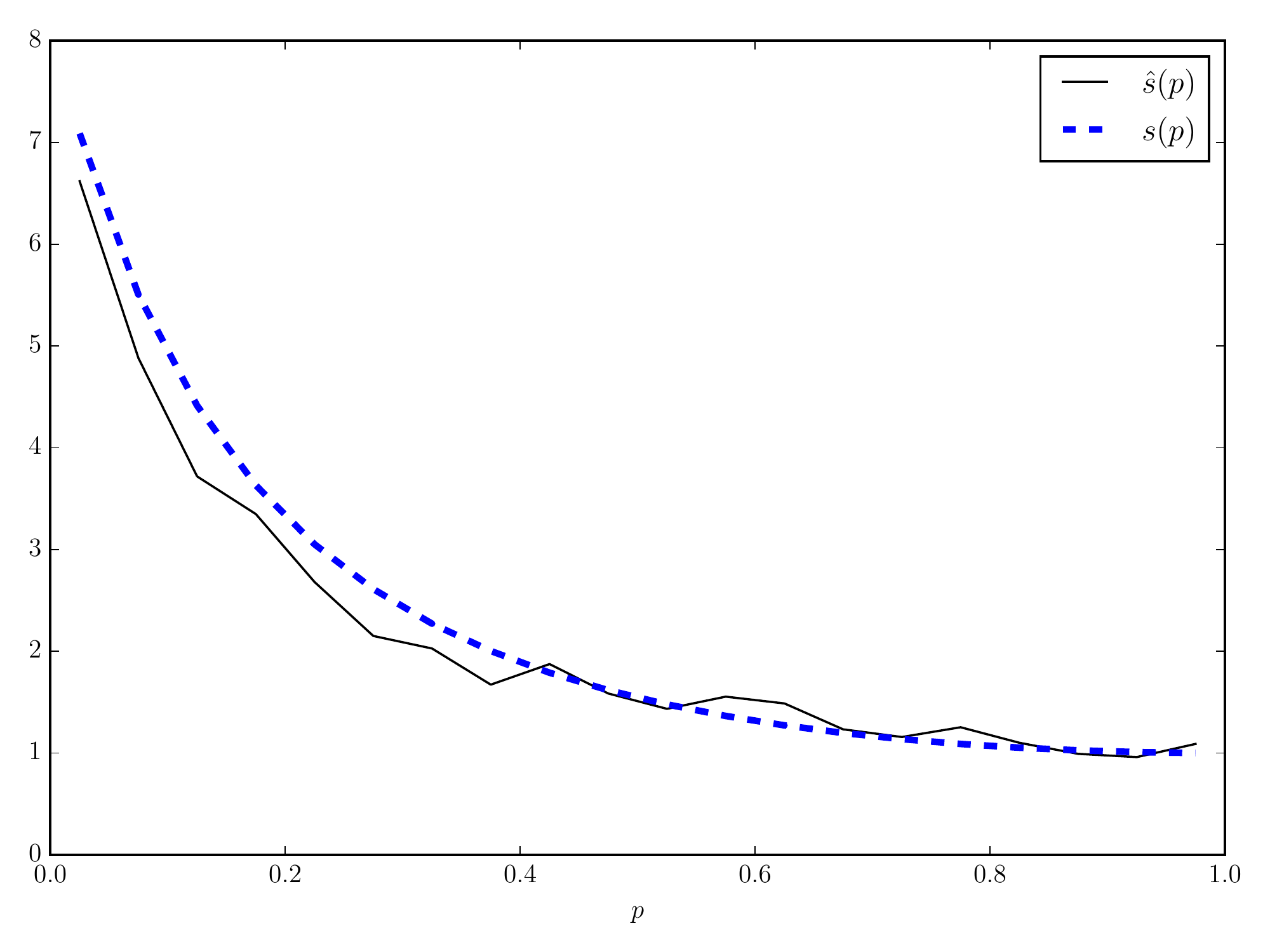}}

  \subfloat[$\hat w(\cdot)$ vs. $w(\cdot)$\label{fig:w}]{%
    \includegraphics[width=0.45\textwidth]{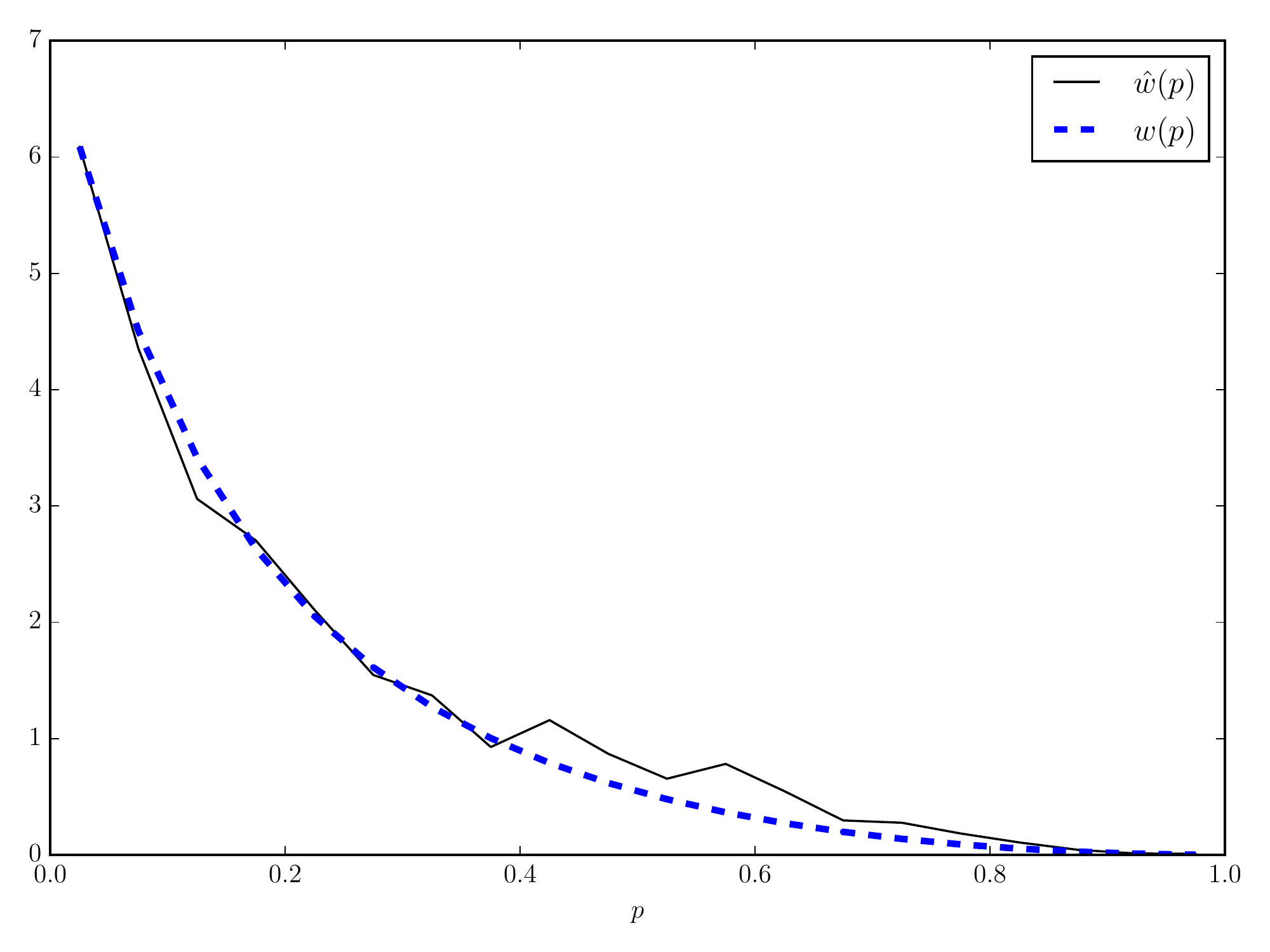}}

  \caption{Estimates of $m(\cdot)$, $s(\cdot)$, and $w(\cdot)$ based
    on the data generated by simulating the system with $c = 2$ and
    $\alpha = 1.5$. For these values of $c$ and $\alpha$, the queue is
    stable and so we use a linear scale for both
    axes.\label{fig:stable}}
\end{figure}


Now consider the unstable case for which $m(\cdot)$, $s(\cdot)$, and
$w(\cdot)$ are finite only for $p \in (p^*, 1] = (0.6, 1]$. As a
result, we do not plot the functions for $p < p^*$. Because of the
vertical asymptote at $p^*$, we use a log-scale for the vertical
axis. The results are plotted in Fig.~\ref{fig:unstable}. In
Fig.~\ref{fig:m_unstable}, we see that $\hat m(\cdot)$ and $m(\cdot)$
seem to agree on $(p^*, 1]$. Fig.~\ref{fig:m_unstable} also depicts
the bifurcation at $p^*$.  We see that $\hat m(p^* - \delta/2)$ is
roughly 10 times the value of $\hat m(p^* + \delta/2)$. This reflects
the fact that $\hat m(p)$ will diverge to infinity as $T \uparrow
\infty$ for $p < p^*$. We see similar results regarding $\hat
s(\cdot)$ in Fig.~\ref{fig:s_unstable}. For $p \in (p^*, 1]$, $\hat
s(p)$ and $s(p)$ agree. Note that for $p < p^*$, neither $\hat s(p)$
nor $s(p)$ appear on the plot because both quantities are
infinite. Hence, we see that $\hat s(\cdot)$ and $s(\cdot)$ agree for
all $p \in [0, 1]$. We see the same results for $\hat w(\cdot)$: the
estimate agrees with the analytic result where both are finite and
also where both are infinite.

\begin{figure}
  \subfloat[$\hat m(\cdot)$ vs. $m(\cdot)$\label{fig:m_unstable}]{%
    \includegraphics[width=0.45\textwidth]{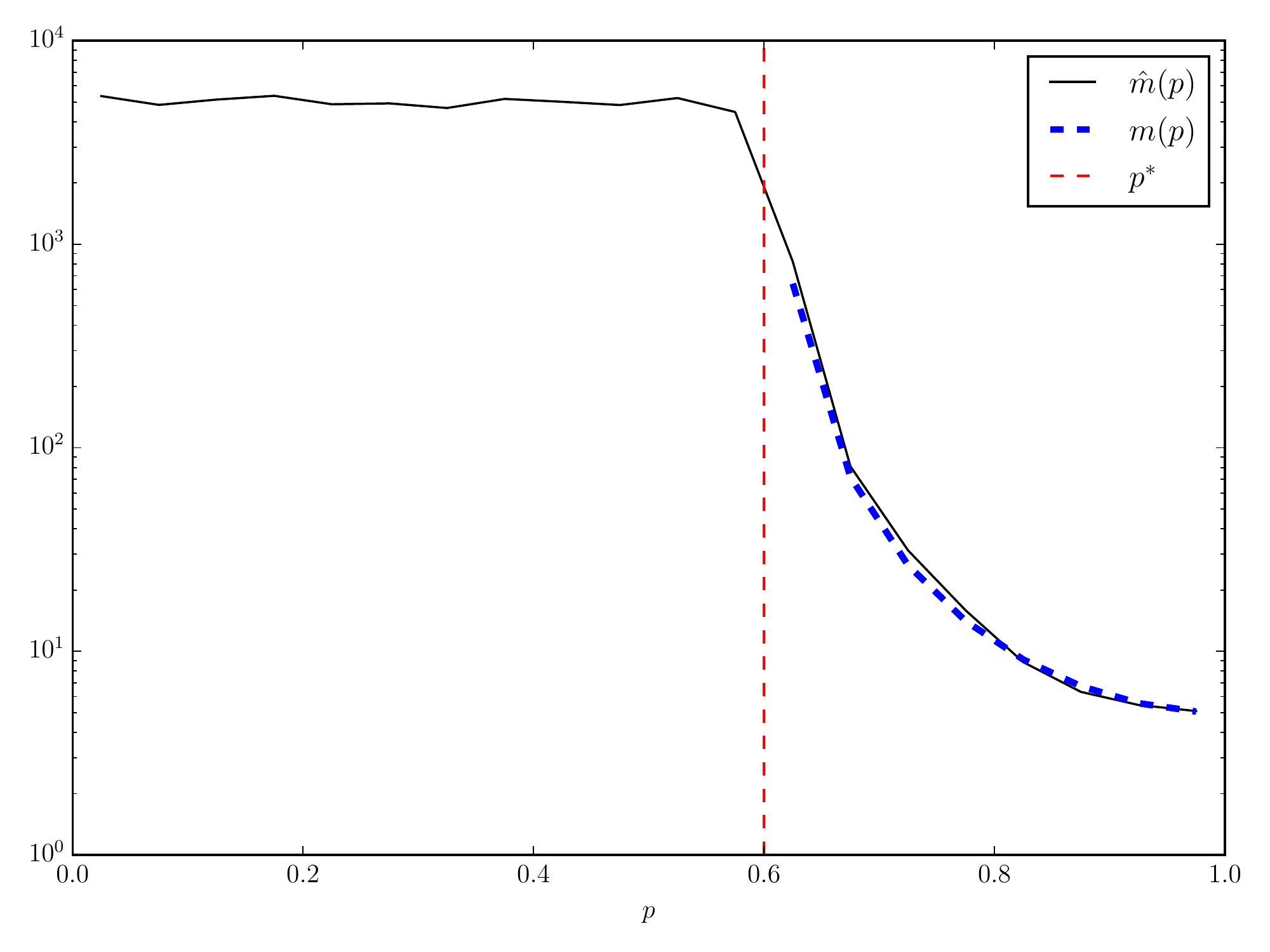}}

  \subfloat[$\hat s(\cdot)$ vs. $s(\cdot)$\label{fig:s_unstable}]{%
    \includegraphics[width=0.45\textwidth]{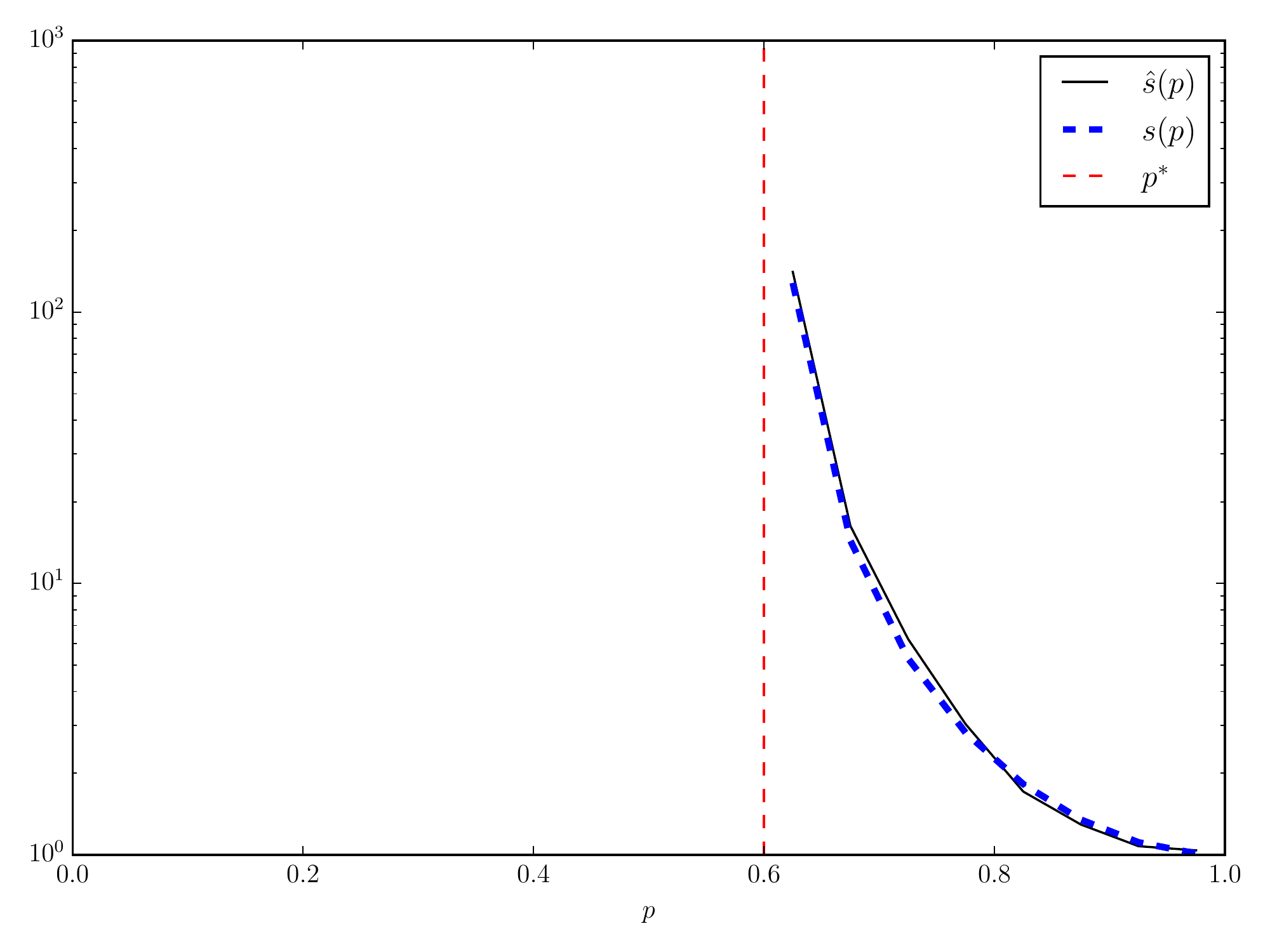}}

  \subfloat[$\hat w(\cdot)$ vs. $w(\cdot)$\label{fig:w_unstable}]{%
    \includegraphics[width=0.45\textwidth]{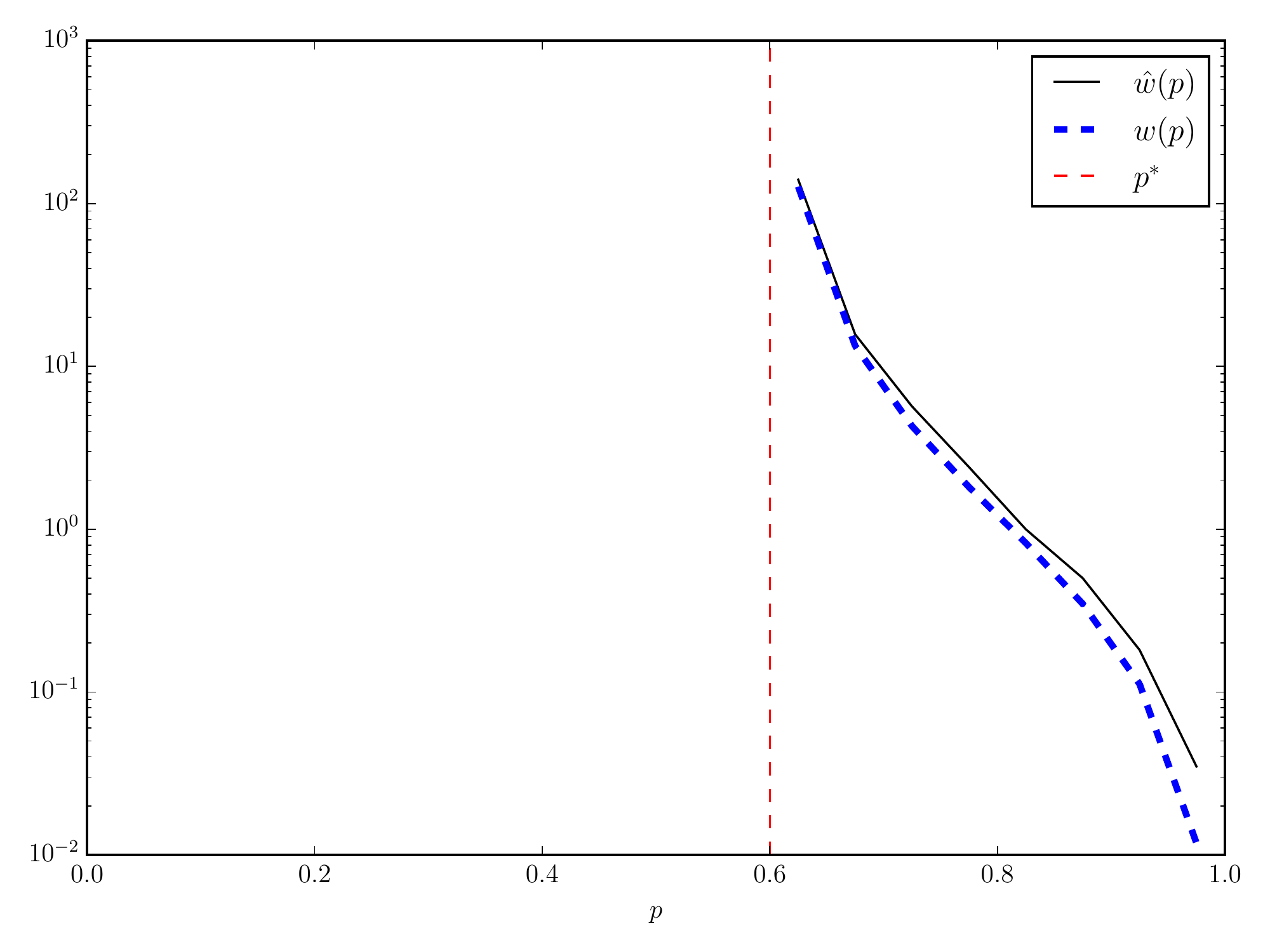}}

  \caption{Estimates of $m(\cdot)$, $s(\cdot)$, and $w(\cdot)$ based
    on the data generated by simulating the system with $c = 2$ and
    $\alpha = 5.0$. Because the queue is unstable, the values of the
    functions become quite large and hence we opt to use a logarithmic
    vertical axis.\label{fig:unstable}}
\end{figure}


\section{Future Work\label{sec:future}}
Our work points to several potential directions of future work. One is
to derive more results about the current model. For example, it would
be interesting to know more about the higher order statistics of
$x(\cdot)$. It would also be interesting to extend this model to a
network setting. With a single queue, $x_t(\cdot)$ is a point measure
on $[0, 1]$ but for a system with $n$ queues we would need to have
$x_t(\cdot)$ be a point measure on $[0, 1]^n$. It seems reasonable to
expect that the steady state distribution would have a product-form as
in Jackson's Theorem~\cite{Jackson_1963}, but the details of the
analysis are not immediately clear. In particular, although the
arriving priority levels are IID $U([0, 1])$ we need to know how
customers' priority levels are correlated after they depart.

As noted in our previous work~\cite{Master_ACC_2017}, it may also be
interesting to consider a heavy traffic analysis. Priority queues are
an example of a system that exhibits ``state-space collapse'' in heavy
traffic~\cite{Reiman_1984}. In brief, we would see that upon
appropriate rescaling, the diffusion limit associated with $X_t(p)$
for $p < p^*$ would be zero. However, it may be possible to consider a
diffusion limit for which $p^* \downarrow 0$ so that the diffusion
limit does not collapse to zero. This idea is not yet well developed
but since our analysis applies to overloaded queues, it may fruitful
to consider.

\section{Conclusions\label{sec:conclusions}}
We have presented an infinite dimensional model for a many server
priority queue in which customers are scheduled preemptively according
to priority levels that are drawn from an continuous probability
distribution. Our steady state analysis characterizes the first-order
statistics of the measure-valued process that describes the priority
levels of the customers in the queue. We have used derived formulae
for the expected sojourn and waiting times of a function of customer
priority level. These results generalize our previous work
\cite{Master_ACC_2017} and contribute to a broader literature on
preemptive scheduling with random priorities
\cite{Haviv_2016}. Discrete event simulations verify our analytical
results and we have discussed some areas of future work.

\bibliographystyle{ieeetr}
\bibliography{NMaster_ZZhou_NBambos_CISS2017}

\end{document}